\documentclass{article}
\usepackage{packages}

\author{Levente Bodnár}
\title{CLIQUE as an AND of Polynomial-Sized Monotone Constant-Depth Circuits} 
\date{February 10, 2022}

\begin{document}

\maketitle

\begin{abstract}
This paper shows that calculating $k$-CLIQUE on $n$ vertex graphs, requires the AND of at least $2^{n/4k}$ monotone, constant-depth, and polynomial-sized circuits, for sufficiently large values of $k$. The proof relies on a new, monotone, one-sided switching lemma, designed for cliques.
\end{abstract}

\section{Introduction} 

Almost all Boolean functions require exponentially large circuits \cite{shannon_synthesis_1949} however, a super-polynomial circuit lower bound on NP complete problems appear out of reach. As lower bounds for general circuits seem difficult to prove, there has been more focus on restricted circuit models. 

A well understood restriction, where there are strong lower bounds, is the class of constant-depth circuits. Early results \cite{Ajtai198311FormulaeOF, FSSparity} proved that PARITY requires super-polynomial size constant-depth circuits. This was subsequently refined in \cite{yao_circuit, hastad_almost_1986}, developing the powerful switching method.

For monotone functions, another fruitful restriction is the class of monotone circuits. Razborov was the first to prove super-polynomial bound on the NP complete problem CLIQUE, for monotone circuits using the method of approximations \cite{razborov_lower_1985}. This was improved by \cite{andreev, AlonBoppana1987, RossmanMonotone} and recently \cite{robust_sunfl} with the current best truly exponential bound $n^{\Omega(k)}$ when $k \leq n^{1/3 - o(1)}.$

\vspace{5mm}

The simplest circuit expressing CLIQUE, is an OR of monomials, one for each possible clique. The top gate has fan-in $\binom{n}{k}$. Alternatively, one can express CLIQUE as monotone AND of monomials, one for each maximal graph without $k$-clique. The number of maximal graphs without triangles was investigated in \cite{triangle_free}, where they show that their number is asymptotically $2^{n^2/8 + o(n)}$. The same paper (in remark 7) constructs $2^{\frac{1-1/(k-1)+o(1)}{4}n^2}$ many maximal graphs without $k$-clique, but method only works for small $k < \log_{4/3}(n)$. Note that any maximal graph without $k$-clique can be extended to a maximal graph without $k'$-clique, by simply adding $k'-k$ points connected to everything. This gives a $2^{\frac{1-1/\log_{4/3}(n-k) + o(1)}{4}(n-k)^2}$ lower bound on the number of maximal, $k$
-clique free graphs, when $k>\log_{4/3}(n)$. Correspondingly, there is a monotone, depth-2 circuit, with top AND gate, calculating $k$-CLIQUE, with top fan-in lower bounded by the same value. This paper is a small progress towards answering the natural question, whether CLIQUE can be expressed as the AND of a smaller number, monotone, constant-depth, and polynomial-sized circuits.

\begin{theorem}\label{main_theorem}
For any $c_{d}, c_{s}$ constants, large enough $n$ and $\log(n)^{c_{d}+6} < k,$ one cannot have less than $2^{n/4k}$ many monotone circuits $\{f_j\}$ on $\binom{n}{2}$ input bits, each with depth and size bounded by $c_{d}$ and $n^{c_{s}}$ respectively, such that $\bigwedge_j f_j$ computes $k$-CLIQUE.
\end{theorem}

The bound looks odd (and surprisingly strong for small $k$) compared to classical circuit bounds for CLIQUE due to the restricted setting. This shows that OR gates are much better at expressing CLIQUE than AND gates in this sense. The result follows from the main lemma of this paper, that monotone constant-depth polynomial-sized circuits with zero-error on YES instances have $2^{-n/4k}$ correlation with cliques on NO instances. This gives a strong asymmetric result; note that a single monomial, stating that a certain $k$ subset is a clique, has zero error on NO instances, but on YES instances, has $\binom{n}{k}^{-1}$ correlation.

The main novelty of the proof is an approximation of a monotone CNF circuit having larger fan-in, by a monotone DNF circuit having a small fan-in, where this approximation is correct on most of the relevant cliques. Approximation of Boolean functions with depth 2 circuits was initially investigated by \cite{approx1}. Continuing this line of work, \cite{approx3} showed strong one-sided bounds for all monotone functions.

\vspace{5mm}

Section \ref{sec_notation} contains the main notation used in this paper, section \ref{sec_outline} covers the main definitions, lemmas and an overall view of the proof. Section \ref{sec_proofs} contains all the technical steps and proofs of lemmas and theorems. The paper finishes with a few concluding thoughts and open problems.

\section{Notation}\label{sec_notation}

For positive integer $n$, use $[n]$ for the set $\{0, 1, ..., n-1\}$. For a set $U$ and integer $m$ write $\binom{U}{m} = \{A \subseteq U : |A|=m\}$. When $m=2$ write $K(U)$ instead $\binom{U}{2}$. For the power set write $[2]^U = \{A \subseteq U\}$.

Use $A, B$ for subsets of $[n]$. Similarly $G, H$ for subsets of $K([n])$ and identify them with graphs. $k$ is a number that can depend on $n$, and use $X, Y, Z$ for subsets of $\binom{[n]}{k}$. The curly versions represent families, so $\mathcal{A}, \mathcal{B}$ are used for families of subsets of $[n]$. Similarly $\mathcal{G}, \mathcal{H}$ are families of subsets of $K([n])$. For a family $\mathcal{A}$ write $\mathcal{K}(\mathcal{A}) = \{K(A) : A \in \mathcal{A}\}$.

Monotone circuits are represented using the letters $f$, $g$, and $h$. Given a circuit with input set $U$, for any $A \subseteq U$ write $f(A)$ for the value of the circuit on input $A$. Therefore they can be used as functions $f : [2]^U \rightarrow [2]$. For convenience use $f=g$ to mean that $f$ and $g$ compute the same function, not that the circuits are the same. The depth $d(f)$ of a circuit is the longest path from the output node to any of the input nodes. The size $|f|$ is the number of nodes in the circuit. For $f, g$ on the same input set $U$, write $f \leq g$ if $\forall A \subseteq U \ f(A) \leq g(A)$. For given $\mathcal{A} \subseteq [2]^U$ the DNF with monomials defined by $\mathcal{A}$ is $i_\mathcal{A} : [2]^U \rightarrow [2]$. This gives the function $$i_{\mathcal{A}}(B) \mapsto \begin{cases} 1 & \text{if \ } \exists A \in \mathcal{A} \  (A \subseteq B) \\ 0 & \text{otherwise}. \end{cases}$$

The letters $\rho, \sigma$ are monotone restrictions. On an input set $U$ they are functions $\rho : U \rightarrow \{1, \star\}$. Given a restriction $\rho$ and a circuit $f$ the restricted circuit $f_\rho : [2]^{\rho^{-1}(\star)} \rightarrow [2]$ is the circuit where every input source from $\rho^{-1}(1)$ is set to constant true. Given an input set $U$, use $R_U^p$ as a distribution on monotone restrictions on $U$, that maps $\rho(u)$ to $\star$ with probability $p$ (and therefore maps to $1$ with probability $1-p$) independently for each $u \in U$.

\vspace{5mm}

\section{Outline}\label{sec_outline} The result follows from the following proposition on a single bounded depth and size monotone circuit. 

\begin{proposition}\label{correl_prop}
For any $c_{d}, c_{s}$ constants, large enough $n$ and any monotone circuit $f$ on $\binom{n}{2}$ input bits, with depth $d(f) \leq c_{d}$ and size $|f| \leq n^{c_{s}}$ such that $f \geq i_{\mathcal{K}\left(\binom{[n]}{k}\right)}$, $f$ must also satisfy $$\underset{G \sim \operatorname{ER}(n, 1-p)}{\mathbb{P}}(f(G)) \geq 1-2^{-n/3k}$$ where $$p = c(c_s, c_d) \log(n)^{-c_d - 4}.$$
\end{proposition} 

\vspace{5mm}

The method is a switching of the bottom layers between small ($O(\log(n)^2)$) fan-in DNF and medium ($O(n/k)$) fan-in CNF. The rationale behind the numbers is the desire to ensure that the circuits remain ineffective at approximating CLIQUE. This holds true intuitively for CNFs even when larger fan-ins are allowed. However, for DNFs, a stronger fan-in restriction is necessary. To express this correlation with CLIQUE, the following approach will be employed:

\begin{definition}
For a monotone restriction $\rho : K([n]) \rightarrow \{1, \star \}$ and a monotone circuit $f:[2]^{K([n])} \rightarrow [2]$ both on graphs, write $Z_\rho(f)$ for the maximal collection of cliques implying the function. It is the largest $Z_\rho(f) \subseteq \binom{[n]}{k}$ where it holds that $\left(i_{\mathcal{K}\left(Z_\rho(f)\right)}\right)_\rho \leq f_\rho$. 
\end{definition}

The main Lemma of the paper shows that the CNF to DNF switching reduces the possible cliques forcing the circuit to true (the set $Z(f)$) by a negligible amount.
\begin{lemma}[CNF to DNF switching with small clique error]\label{hard_lemma}
For any $f:[2]^{K([n])} \rightarrow [2]$ monotone $s$-CNF and a monotone restriction $\rho:K([n]) \rightarrow \{1, \star\}$, one can find $g$ monotone $t$-DNF, that satisfies $g_\rho \leq f_\rho$ and $$\left| Z_\rho(f) \setminus Z_\rho(g) \right| \leq \binom{n}{k}\left(s \frac{k}{n}\right)^{\sqrt{t/2}}.$$
\end{lemma}

The DNF to CNF switching requires a random restriction, but after the restriction the functions are the same

\begin{lemma}[DNF to CNF switching]\label{easy_lemma}
For any $g$ monotone $t$-DNF after taking $\rho \sim R_U^{1/(2t)}$ restriction $g$ can be written as an equivalent $s$-CNF with probability $ \geq 1 - 2^{-s-1}$.
\end{lemma}

The end result is a $O(\log(n)^2)$ fan-in DNF that is still true on a large number of cliques, therefore it holds that a random 0-1 assignment on the remaining variables satisfies the circuit with high probability. The probability that one of the DNF to CNF switching fails or that the final random assignment does not force the circuit to true, will be exponentially small, allowing the large fan-in bound on the output gate.

\vspace{5mm}

\section{Proofs}\label{sec_proofs}

The following is a list of simple observations, the proof is not included.

\begin{observation}\label{observs}\leavevmode
\begin{enumerate}
    \item If $\mathcal{G} \subseteq \mathcal{H}$ then $i_{\mathcal{G}} \leq i_{\mathcal{H}}$.
    \item If $f \leq g$ then $Z(f) \subseteq Z(g)$.
    \item If $\rho \sim R_U^p$ and $\sigma \sim R_{\rho^{-1}(\star)}^q$ then $\sigma \circ \rho \sim R_U^{pq}$.
    \item If $G = \rho^{-1}(1)$ where $\rho \sim R_{K([n])}^p$ then $G \sim \operatorname{ER}(n, 1-p)$.
\end{enumerate}
\end{observation}

The idea for Lemma \ref{hard_lemma} is to build a monotone decision tree both by cliques and by edges simultaneously based on the formula. This gives enough control over the cliques they imply. The extension of edge sets to cliques can only decrease the function but the clique implications $Z(f)$ will not change. The tree is pruned to include only clauses below the cut. This can be achieved with a minimal loss of cliques.

\begin{proof}[Proof of Lemma \ref{hard_lemma}]
First simplify the function $f$ to only include edges from $\rho^{-1}(\star)$. If any of the clauses become trivial, then $f_{\rho} \equiv 1$ and can set $g \equiv 1$. Let's name the clauses $f_{\rho} = \bigwedge_{j=1}^m \left( \bigvee_{e \in q_j} e\right)$ where $|q_j| \leq s$. 

Create two trees $T \subseteq T'$. $T$ has its nodes labeled by subsets of $[n]$, while $T'$ has the labels from subsets of $K([n])$. Furthermore associate to every non-leaf a clause from $f_{\rho}$ both in $T$ and $T'$. For $v \in T$ denote $A(v) \subseteq [n]$ the label and $q_T(v)$ the associated clause. Similarly denote $G(w) \subseteq K([n])$ the label and $q_{T'}(w)$ for $w \in T'$ the clause. Construct the labels such that for $v \in T \subseteq T'$ it holds that $\bigcup G(v) = A(v)$ meaning in particular that $G(v) \subseteq K(A(v))$. Also for $v \in T$ not a leaf in $T$ the labels will have $q_T(v)=q_{T'}(v)$.

The root $r \in T$ and $r \in T'$ has $\emptyset = A(r) = G(r)$. A node $v \in T$ is a leaf if $f_\rho(K(A(v)))=1$, similarly $w \in T'$ is a leaf if $f_\rho(G(w))$. If $v \in T$ is not a leaf (therefore not a leaf in $T'$) find the clauses in $f_\rho$ not satisfied by $K(A(v))$. Let the first clause be $q_j$, then assign $q_T(v)=q_{T'}(v)=q_j$. If $w$ is a leaf in $T$ but not in $T'$ or is not a member of $T$ at all, then choose $q_j$ to be the first clause not satisfied by $G(w)$ and set $q_{T'}(w)=q_j$. Then for $v \in T$ with a clause label $q_T(v)=q_{T'}(v)$ and for each edge $e \in q_T(v)$ add $v'$ a child of $v$ both in $T$ and $T'$ with label $A(v')=A(v) \cup e$ and $G(v')=G(v) \cup \{e\}$. This preserves that $\bigcup G(v') = A(v')$. For $w \in T' \setminus T$ not a leaf, only add the $w'$ children and $G(w')$ labels without any constrains from $T$.

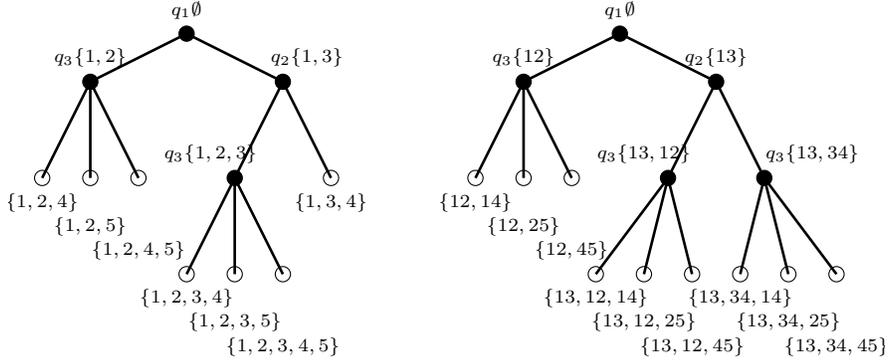
\begin{figure}[ht]
\begin{tikzpicture}[scale = 0.64][line cap=round,line join=round,>=triangle 45,x=1cm,y=1cm]
\clip(-0.7,-2) rectangle (17.5,9);
\draw [line width=1pt] (3,5)-- (1,4);
\draw [line width=1pt] (3,5)-- (5,4);
\draw [line width=1pt] (4,2)-- (5,4);
\draw [line width=1pt] (5,4)-- (6,2);
\draw [line width=1pt] (4,2)-- (3,0);
\draw [line width=1pt] (4,2)-- (4,0);
\draw [line width=1pt] (4,2)-- (5,0);
\draw [line width=1pt] (1,4)-- (0,2);
\draw [line width=1pt] (1,4)-- (1,2);
\draw [line width=1pt] (1,4)-- (2,2);
\draw [line width=1pt] (11.5,0)-- (13,2);
\draw [line width=1pt] (12.5,0)-- (13,2);
\draw [line width=1pt] (13,2)-- (13.5,0);
\draw [line width=1pt] (14.5,0)-- (15,2);
\draw [line width=1pt] (15,2)-- (15.5,0);
\draw [line width=1pt] (16.5,0)-- (15,2);
\draw [line width=1pt] (9,2)-- (10,4);
\draw [line width=1pt] (10,2)-- (10,4);
\draw [line width=1pt] (10,4)-- (11,2);
\draw [line width=1pt] (10,4)-- (12,5);
\draw [line width=1pt] (14,4)-- (12,5);
\draw [line width=1pt] (14,4)-- (13,2);
\draw [line width=1pt] (15,2)-- (14,4);
\begin{scriptsize}
\draw [fill=black] (3,5) circle (4.5pt);
\draw[color=black] (3,5.5) node {$q_1 \emptyset$};
\draw [fill=black] (1,4) circle (4.5pt);
\draw[color=black] (1,4.5) node {$q_3 \{1, 2\}$};
\draw [fill=black] (5,4) circle (4.5pt);
\draw[color=black] (5.5,4.5) node {$q_2 \{1, 3\}$};
\draw [color=black] (0,2) circle (4.5pt);
\draw[color=black] (0,1.5) node {$\{1, 2, 4\}$};
\draw [color=black] (1,2) circle (4.5pt);
\draw[color=black] (1,1) node {$\{1, 2, 5\}$};
\draw [color=black] (2,2) circle (4.5pt);
\draw[color=black] (2,0.5) node {$\{1, 2, 4, 5\}$};
\draw [fill=black] (4,2) circle (4.5pt);
\draw[color=black] (3.5,2.5) node {$q_3 \{1, 2, 3\}$};
\draw [color=black] (6,2) circle (4.5pt);
\draw[color=black] (6,1.5) node {$\{1, 3, 4\}$};
\draw [color=black] (3,0) circle (4.5pt);
\draw[color=black] (3,-0.5) node {$\{1, 2, 3, 4\}$};
\draw [color=black] (4,0) circle (4.5pt);
\draw[color=black] (4,-1) node {$\{1, 2, 3, 5\}$};
\draw [color=black] (5,0) circle (4.5pt);
\draw[color=black] (5,-1.5) node {$\{1, 2, 3, 4, 5\}$};
\draw [fill=black] (12,5) circle (4.5pt);
\draw[color=black] (12,5.5) node {$q_1 \emptyset$};
\draw [fill=black] (10,4) circle (4.5pt);
\draw[color=black] (10,4.5) node {$q_3 \{12\}$};
\draw [fill=black] (14,4) circle (4.5pt);
\draw[color=black] (14,4.5) node {$q_2 \{13\}$};
\draw [fill=black] (13,2) circle (4.5pt);
\draw[color=black] (12.5,2.5) node {$q_3 \{13, 12\}$};
\draw [fill=black] (15,2) circle (4.5pt);
\draw[color=black] (16,2.5) node {$q_3 \{13, 34\}$};
\draw [color=black] (9,2) circle (4.5pt);
\draw[color=black] (9,1.5) node {$\{12, 14\}$};
\draw [color=black] (10,2) circle (4.5pt);
\draw[color=black] (10,1) node {$\{12, 25\}$};
\draw [color=black] (11,2) circle (4.5pt);
\draw[color=black] (11,0.5) node {$\{12, 45\}$};
\draw [color=black] (13.5,0) circle (4.5pt);
\draw[color=black] (13.5,-1.5) node {$\{13, 12, 45\}$};
\draw [color=black] (12.5,0) circle (4.5pt);
\draw[color=black] (12.5,-1) node {$\{13, 12, 25\}$};
\draw [color=black] (11.5,0) circle (4.5pt);
\draw[color=black] (11.5,-0.5) node {$\{13, 12, 14\}$};
\draw [color=black] (14.5,0) circle (4.5pt);
\draw[color=black] (14.5,-0.5) node {$\{13, 34, 14\}$};
\draw [color=black] (15.5,0) circle (4.5pt);
\draw[color=black] (15.5,-1) node {$\{13, 34, 25\}$};
\draw [color=black] (16.5,0) circle (4.5pt);
\draw[color=black] (16.5,-1.5) node {$\{13, 34, 45\}$};
\end{scriptsize}
\end{tikzpicture}
\captionsetup{labelformat=empty}
\centering
\caption{The trees on formula $f= q_1 \wedge q_2 \wedge q_3 = (12 \vee 13) \wedge (12 \vee 34) \wedge (14 \vee 25 \vee 45)$ with clauses numbered accordingly.}
\end{figure}

Call a node's depth its distance from root. Let $\mathcal{A}_d$ be the collection of $A(v)$ where $v$ is a leaf with depth $\leq d$ in $T$. Let $\mathcal{B}_d$ be the collection of $A(v)$ where $v$ is any node in $T$ with depth exactly $d$ (not necessarily a leaf). Similarly define $\mathcal{G}_d$ to be the collection of $G(w)$ for $w \in T'$ leaf node of depth $\leq d$. Let $d(T), d(T')$ be the maximal depth of the trees.

\begin{claim}[Relations involving the $\mathcal{G}_d, \mathcal{A}_d, \mathcal{B}_d$ sets]\label{hard_lemma_claims} \leavevmode
\begin{enumerate}
    \item $i_{\mathcal{G}_{d(T')}} = f_\rho$
    \item $i_{\mathcal{K}(\mathcal{A}_{d(T)})} \leq f_\rho$
    \item $Z_\rho \left(i_{\mathcal{K}(\mathcal{A}_{d(T)})}\right) = Z_{\rho}(f)$
    \item $i_{\mathcal{K}\left(\mathcal{A}_{d(T)}\right)} \leq i_{\mathcal{K}\left(\mathcal{A}_{d} \cup \mathcal{B}_{d+1}\right)}$
    \item $\left| \mathcal{B}_d \right| \leq s^d$
\end{enumerate}
\end{claim}

\vspace{5mm}

\begin{proof}[Proof of Claim \ref{hard_lemma_claims}] \leavevmode
\begin{enumerate}
    \item Working in $U = \rho^{-1}(\star)$, for any $u \subseteq U$ it holds that $f_\rho(u)$ iff for all $q_j$ it is true that $u \cap q_j \neq \emptyset$. Therefore every non leaf $w \in T'$ has an edge $e_w \in q_{T'}(w)$ such that $e_w \in u$. Following the label set increments $G(w)$ to $G(w) \cup \{e_w\}$, provides a path from root to a leaf $w$ such that each label in that path is a subset of $u$ resulting in $G(w) \subseteq u$ and therefore $i_{\mathcal{G}_{d(T')}}(u)=1$. This gives $i_{\mathcal{G}_{d(T')}} \geq f_\rho$. For the other direction take $u$ such that $i_{\mathcal{G}_{d(T')}}(u)=1$. Then there is a leaf $w \in T$ where $G(w) \subseteq u$. By definition of a leaf, $G(w)$ already satisfies $f_\rho$ giving that by monotonicity $f_\rho(u)=1$.
    \item Consider any $u \subseteq U$ such that $i_{\mathcal{K}\left(\mathcal{A}_{d(T)}\right)}(u)=1$, then there must be a leaf $v \in T$ and a corresponding $A(v) \in \mathcal{A}_{d(T)}$ with $K(A(v)) \subseteq u$. Then as $v$ is a leaf $K(A(v))$ satisfies $f_\rho$ and by monotonicity so does $u$.
    \item Using observation \ref{observs} and the previous point it holds that $Z_\rho \left( i_{\mathcal{K}\left(\mathcal{A}_{d(T)}\right)}\right) \subseteq Z_{\rho}(f)$. For the other direction, take any $B \in Z_\rho(f)$. Then $K(B)$ satisfies $f_\rho$ which by above agrees with $i_{\mathcal{G}_{d(T')}}$. Take a leaf $w \in T'$ where $G(w) \subseteq K(B)$. Then as $T \subseteq T'$ the path from the root to $w$ must contain a leaf from $T$, say it is $v$. Then $G(v) \subseteq G(w) \subseteq K(B)$ therefore $A(v) = \bigcup G(v) \subseteq \bigcup G(w) \subseteq B$ so $K(B)$ satisfies $i_{\mathcal{K}\left(\mathcal{A}_{d(T)}\right)}$ as required.
    \item Suppose for $u \subseteq U$ it also holds that $i_{\mathcal{K}(\mathcal{A_{d(T)}})}(u) = 1$, then there must be a leaf $v \in T$ and a corresponding $A(v) \in \mathcal{A}_{d(T)}$ with $K(A(v)) \subseteq u$. If the depth of $v$ is $\leq d$ then $v \in \mathcal{A}_d$ and $i_{\mathcal{K}(\mathcal{A}_d)}(u)=1$. Otherwise there must be a path from root to $v$ going through a node at depth $d+1$, giving that $i_{\mathcal{K}(\mathcal{B}_{d+1})}(u)=1$. Using observation \ref{observs}, the claim follows.
    \item Every set in $\mathcal{B}_d$ can be identified with a path from the root to a node at level $d$. As the tree branches at each step to at most $s$ new nodes, the mentioned bound follows.
\end{enumerate}
\end{proof}

\vspace{5mm}

Combining the above with observation \ref{observs} and by noting for all $d$ it holds that $\mathcal{A}_d \subseteq \mathcal{A}_{d(T)}$, $$Z_\rho \left(i_{\mathcal{K}\left(\mathcal{A}_{d}\right)}\right) \subseteq Z_\rho (f) \subseteq Z_\rho \left(i_{\mathcal{K}\left(\mathcal{A}_{d} \cup \mathcal{B}_{d+1}\right)}\right)$$ this gives $$\left| Z_\rho(f) \setminus Z_\rho \left(i_{\mathcal{K}\left(\mathcal{A}_{d}\right)}\right) \right| \leq  \left| Z_\rho \left(i_{\mathcal{K}\left(\mathcal{A}_{d} \cup \mathcal{B}_{d+1}\right)} \right) \setminus Z_\rho \left(i_{\mathcal{K}\left(\mathcal{A}_{d}\right)}\right) \right| \leq \sum_{B \in \mathcal{B}_{d+1}} \binom{n-|B|}{k-|B|}$$ note that every $B \in \mathcal{B}_{d+1}$ is constructed in $d+1$ steps, each increasing the cardinality by at least $1$ but at most $2$, therefore $d+1 \leq |B| \leq 2d+2$ meaning $$ \left| Z_\rho(f) \setminus Z_\rho \left(i_{\mathcal{K}\left(\mathcal{A}_{d}\right)}\right) \right| \leq s^{d+1} \binom{n-(d+1)}{k-(d+1)} \leq \binom{n}{k}\left(s \frac{k}{n} \right)^{d+1}$$ as in $i_{\mathcal{K}\left(\mathcal{A}_d\right)}$ each clause has at most $\binom{2d}{2}$ edges, setting $d$ such that $t = \binom{2d}{2}$ gives the result with $g = i_{\mathcal{K}\left(\mathcal{A}_d\right)}$.

\end{proof}

\vspace{5mm}

\begin{proof}[Proof of Lemma \ref{easy_lemma}]
Build a tree $T'$ the same way as in Lemma \ref{hard_lemma} (without caring about any $T$) using the clauses of $g$. Write $\mathcal{H}_d$ for the collection of $G(w)$ sets where $w$ is a node at depth exactly $d$. Using this tree, write $g_{\rho}$ as a $s$-CNF if all the $\mathcal{H}_{s+1}$ OR clauses are satisfied by $\rho$. Similar to Claim 7  there is $|\mathcal{H}_{s + 1}| \leq t^{s+1}$ and each $H \in \mathcal{H}_{s+1}$ has exactly $s+1$ edges. Therefore a $\rho \sim R_U^{1/(2t)}$ satisfies all the OR clauses in $\mathcal{H}_{s+1}$ with probability lower bounded by $$1 - |\mathcal{H}_{s + 1}|\left( \frac{1}{2t}\right)^{s+1} \geq 1-2^{-s-1}$$ as needed.
\end{proof}

Then proposition \ref{correl_prop} follows from the combination of the two switchings. 
\begin{proof}[Proof of Proposition \ref{correl_prop}]
First transform the circuit to having alternating AND and OR layers of gates and extend the bottom with dummy gates to have fanin $1$. This results in $f'$ agreeing with $f$ on all inputs, while still $|f'| \leq n^{c_{s}+1}$ and $d(f') \leq c_{d}+1$. Fix $t = 2(c_s+2)^2\log(n)^2$ and $s=n/2k$ this is to ensure that $\left(s \frac{k}{n}\right)^{\sqrt{t/2}} = n^{-(c_s+2)}$ The proof relies on the changing of the bottom layer between $s$-CNFs and $t$-DNFs, after each switch, merging with the layer above therefore reducing the depth. 

Provided the bottom layer is a $t$-DNF, use Lemma \ref{easy_lemma} for each DNF, introducing at depth $i$ a random $\rho_j \sim R^{1/(2t)}_{\rho^{-1}_{(<i)}(\star)}$ where write $\rho_{(<i)} = \underset{j<i}{\circ} \rho_j$ for the composition of previous restrictions. The change is unsuccessful with probability $\leq 2^{-s-1}$ at each DNF, otherwise results in a $s$-CNF. 

Provided the bottom layer is a $s$-CNF use Lemma \ref{hard_lemma}. Each gate $g$ in the bottom CNF layer is replaced with $g'$ a $t$-DNF satisfying $g' \leq g$ and by the selection of parameters, it holds that $$\left| Z_{\rho_{(\leq i)}}(g) \setminus Z_{\rho_{(\leq i)}}(g') \right| \leq \binom{n}{k}\left(s \frac{k}{n}\right)^{\sqrt{t/2}} = \binom{n}{k}n^{-c_{s}-2}.$$

The probability that any of the DNF to CNF switching fails is $\leq n^{c_{s}+1}2^{-s-1}$ by the union bound. Write $\rho = \rho_{(<c_{d}+2)}$. The total number of cliques lost from CNF to DNF switches is at most $\binom{n}{k}n^{-1}$ from the union bound. This gives that with probability $\geq 1-n^{c_{s}+1}2^{-s-1}$ it is possible to find $f''$ a $t$-DNF such that $(f'')_\rho \leq (f')_\rho$ and by noting $Z(f') = \binom{[n]}{k}$ it holds that $\left| Z_{\rho}(f'')\right| \geq \binom{n}{k}(1-n^{-1}).$

Suppose $f'' = \bigvee_j \left( \bigwedge_{e \in q_j} e \right).$ To conclude the proof disjoint sets among $q_j$ will be selected. Iteratively pick the sets $r_j$ from $q_j$ such that they are disjoint. Start with $r_1=q_1$ and suppose inductively that $r_1, ..., r_x$ are selected so far and they are all disjoint. Consider $$X = \{A \in Z_{\rho}(f'') : \forall j \leq x,  \ r_j \cap K(A) = \emptyset\}$$ notice that $$|X| \geq \binom{n}{k}\left( 1-n^{-1}-xs\frac{k(k-1)}{n(n-1)} \right)$$ since each $|r_j| \leq s$ and the possible cliques containing one particular edge has size $\binom{n-2}{k-2}$. Therefore each edge in $r_j$ removes this number of cliques at most. Provided $|X| > 0$ it is possible to pick $A \in X$ and as $f''(K(A) \cup \rho^{-1}(1))=1$ there must be one $q_j$ satisfied by the inclusion of $K(A)$ giving that $q_j$ is disjoint from $\{ r_j : j \leq x \}$ and it is possible to extend the collection. Therefore, at least $x=\frac{n(n-1)}{2sk(k-1)}$ many disjoint clauses can be constructed, giving that after assigning each remaining $\rho^{-1}(\star)$ edges to $\star$ with probability $1/2t$ giving an extra $\sigma \sim R^{1/2t}_{\rho^{-1}(\star)}$ restriction, it holds that $$ \mathbb{P}\left(f''\left(\sigma^{-1}(1)\right)\right)  \geq 1-\prod_{j=1}^x \left(1-\left(1-\frac{1}{t}\right)^{|r_j|}\right) \geq 1 - \left(1-\frac{1}{e^2}\right)^{\frac{n(n-1)}{2\log(n)k(k-1)}}. $$ This combined with the probability that any of the $c_{d}$ many switching fails, gives the bound $$\mathbb{P}\left(f\left((\sigma \circ \rho)^{-1}(1)\right)\right) \geq 1- n^{c_{s}+1}2^{-\frac{n}{2k}-1} - \left(1-\frac{1}{e}\right)^{\frac{n(n-1)}{2\log(n)k(k-1)}} \geq 1-2^{-\frac{n}{3k}}.$$ At each restriction a $\sigma, \rho_j \sim R^{1/(2t)}_{U}$ is used. Every second layer uses one such restriction, and an additional restriction used at the last step, giving a total $\sigma \circ \rho \sim R_{K([n])}^{p}$ where $$p \geq \left(\frac{1}{2t}\right)^{c_d/2+2} = \left(2c_s+2\right)^{-c_d-4} \log(n)^{-c_d-4}.$$
\end{proof}

\vspace{5mm}

\begin{proof}[Proof of Theorem \ref{main_theorem}]
Use proposition \ref{correl_prop} with the same $c_{d}$ and $c_{s}$ constants. This gives every AC0 circuit satisfying the constraints is true on a $G \sim R_{K([n])}^{p}$ input with probability $1-2^{-n/3k}$, therefore all the $f_j$ is true on input $G$ with probability $1-o(1)$ using the union bound, but note that for $\log(n)^{c_{d}+6} < k$ the probability that a clique appears in $G$ is upper bounded by $$\binom{n}{k}(1-p)^{\binom{k}{2}} \leq 2^{\log(n)k - \frac{k^2}{c'(c_d, c_s) \log(n)^{c_d + 3}}} = o(1)$$ using $p$ from proposition \ref{correl_prop}. Therefore the circuit can not express clique as there is a nonzero probability that $G$ forces all the $f_j$ to true, but $G$ still has no clique.
\end{proof}

\section{Concluding remarks} As outlined in the introduction, one can write CLIQUE as an OR of $\binom{n}{k}$ many monotone polynomial-sized circuits, with a top OR gate. This gives a $\binom{n}{k}^{-1}$ one-sided correlation between monotone polynomial-sized circuits and clique, which is correct on all negative instances. One can ask the dual question:

\begin{question}
    What is $S(n, k)$, the maximum one-sided correlation between monotone polynomial-sized circuits and $k$-CLIQUE on $n$ vertices, which is correct on all positive instances?
\end{question}

Here, a partial answer is given, a $S(n, k) \leq 2^{-n/4k} $ (for large enough $k$), that is unfortunately far from the best known constructions. It would be interesting to know tighter bounds for $S(n, k)$. 

\begin{conjecture}
$S(n, k) = 2^{-\omega(n)}$ for $k$ in a suitable range (for example $k=n^{\alpha}$ for some $0 < \alpha < 1$).
\end{conjecture}

This would also imply that monotone monadic second-order sentences with universal monadic predicates are not strong enough to express $k$-CLIQUE, while monotone existential second-order sentences can express it. Looking at the proof of Lemma \ref{hard_lemma}, the inequality $$\sum_{B \in \mathcal{B}_{d+1}} \binom{n-|B|}{k-|B|} \leq s^{d+1} \binom{n-(d+1)}{k-(d+1)}$$  is used, following from $d+1 \leq |B| \leq 2d+2$. A better approximation of the set sizes on average could provide a stronger result, in particular if all $|B| \approx 2d+2$, then the resulting bound is $S(n, k) \leq 2^{-O(n^2/k)}$. The following stronger conjecture captures this:

\begin{conjecture}
It holds that $S(n, k) \leq 2^{-\Theta(n^2/k)}$ for $k$ in a suitable range.
\end{conjecture}

Note that while this paper discusses the monotone question, as the method heavily uses that condition, the general question is equally interesting:

\begin{question}
What is the maximum one-sided correlation between polynomial-sized circuits and $k$-CLIQUE on $n$ vertices, which is correct on all positive instances?
\end{question}

\bibliographystyle{alpha}
\bibliography{readings.bib}
\appendix
\end{document}